\documentclass[11pt]{llncs}

\usepackage{times}
\usepackage{amssymb,amsmath,amsfonts}
\usepackage{graphicx,subfigure}
\usepackage{enumerate}
\usepackage{booktabs}
\usepackage{hyperref}
\usepackage{color}
\usepackage{xspace}
\usepackage{wrapfig}
\usepackage{cite}
\usepackage{simplemargins}
\setallmargins{1in}

\graphicspath{{pics/}}

\newcommand{\blink}[1]{\textnormal{\texttt{#1}}}
\newcommand{\NP}[0]{\blink{NP}\xspace}

\newcommand{\df}{\textit}

\newcommand{\R}{\mathbb R}
\newcommand{\I}{{\ensuremath{\mathcal{I}}}}
\newcommand{\F}{{\ensuremath{\mathcal{F}}}}

\newcommand{\rar}{\rightarrow}
\DeclareMathOperator{\mad}{mad}

\let\doendproof\endproof
\renewcommand\endproof{~\hfill\qed\doendproof}

\newtheorem{pp}{}

\title{Weak Unit Disk and Interval Representation of Planar Graphs}
\author{Md.~Jawaherul~~Alam \and Stephen~G.~Kobourov \and Sergey~Pupyrev \and
 Jackson~Toeniskoetter}
\institute{Department of Computer Science, University of Arizona, Tucson, USA}

\begin{document}
\date{}
\pagenumbering{arabic}
\pagestyle{plain}

\maketitle

\begin{abstract}
We study a variant of intersection representations with unit balls, that is, unit disks in the plane and
unit intervals on the line. Given a planar graph and a bipartition of the edges of the graph into
{\em near} and {\em far} sets,
the goal is to represent the vertices of the graph by
unit balls so that the balls representing two adjacent vertices intersect if and only
if the corresponding edge is near.
We consider the problem in the plane and prove that it is NP-hard to decide whether
such a representation exists for a given edge-partition. On the other hand,
every series-parallel graph admits such a representation with unit disks for any
near/far labeling of the edges.
We also show that the representation problem on the line is equivalent to a variant of
a graph coloring. We give examples of girth-4 planar and girth-3 outerplanar graphs that
have no such representation with unit intervals. On the other hand,
all triangle-free outerplanar graphs and all graphs with maximum average
degree less than 26/11 can always be represented. In particular,
this gives a simple proof of representability of all planar graphs with large
girth.
\end{abstract}

\section{Introduction}
\label{sect:intro}

Intersection graphs of various geometric objects have been extensively studied
for their many applications~\cite{HK01}. A graph is a $d$-dimensional \df{unit ball graph} if its
vertices are represented by unit size balls in $\R^d$, and an edge exists between two
vertices if and only if the corresponding balls intersect.
Unit ball graphs are called \df{unit disk graphs} when $d=2$ and
\df{unit interval graphs} when $d=1$.
In this paper we study the so called \df{weak unit ball graphs}. Given
a graph $G$ whose edges have been partitioned into ``near'' and ``far'' sets,
we wish to assign unit balls to the vertices of $G$ so that, for an edge
$(u,v)$ of $G$, the balls representing $u$ and $v$ intersect if and only
if the edge $(u,v)$ is near. Note that if $(u,v)$ is not an edge of $G$, then the
balls of $u$ and $v$ may or may not intersect. We refer to such graphs
as \df{weak unit disk} or \df{weak unit interval graphs} when
$d=2$ or $d=1$, respectively. A geometric representation of such graphs
(particularly, a mapping of the vertices to $\R^2$ or $\R$),
is called a \df{weak unit disk drawing} or a \df{weak unit interval drawing}; see
Fig.~\ref{fig:disks}. In figures near edges are shown as
thick line segments and far edges are dashed line segments.
Unit disk drawings allow us to represent the edges of a graph by proximity, which is intuitive from human perception point of view. Weak unit disk graphs also allow to arbitrarily forbid
edges between certain pairs of vertices, which may be useful in drawing ``almost'' unit disk graphs.
It has been shown that weak unit interval graphs can be used to compute
\emph{unit-cube contact representations} of planar graphs~\cite{Bremner12};
see Appendix~\ref{app:cubes} for more details.

\begin{figure}[t]
	\centering
	\subfigure[][]{
		\includegraphics[scale=0.75,page=3]{pics/motivation}
	}
~~~~~~~~~~~~~~~~~~~~~~~~~
	\subfigure[][]{
		\includegraphics[scale=0.75,page=1]{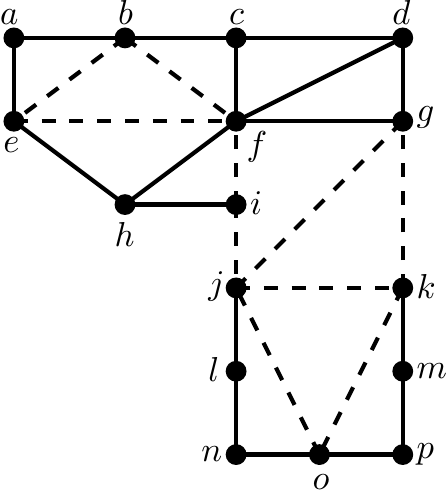}
	}
	\subfigure[][]{
		\includegraphics[scale=0.75,page=2]{pics/motivation}
	}
	\caption{
    (a)~A planar graph with an edge-labeling and its weak unit interval representation.
    (b-c)~A graph with its weak unit disk representation.
    In the figures we indicate near edges with solid lines and far edges with dashed lines.}
	\label{fig:disks}
\end{figure}

Unit disk graphs have been extensively studied for their application to wireless sensor
and radio networks. In such a network, we can model each sensor or radio as a device
with a unit size broadcast range, which naturally induces a unit disk graph by adding
an edge whenever two ranges intersect (or equivalently, when the range of one node contains
a second node). This setting makes it easy to study various practical problems. For example,
in the \emph{frequency assignment problem}, we wish to assign frequencies to radio towers
so that nearby towers do not interfere with each other.
A weakness of the unit disk model is that it does not allow
for interference between nodes from weather and geography, and it does not account for the
possibility that a pair of nodes may not be able to communicate due to technological barriers. One
attempt to address this issue has been the \emph{quasi unit disk graph}~\cite{KWZ03}, where each vertex
is represented by a pair of disks, one of radius $r$, $0<r<1$, and the other of radius 1. In this
model, two vertices are connected by an edge if their radius-$r$ disks overlap, and do not have an
edge if their radius-1 disks do not overlap. The remaining edges are in or out of the graph on a
case by case basis. In the weak unit disk model, such problems can be dealt with by simply deleting edges
between nodes which are nearby but nevertheless whose ranges do not overlap (for example because they are separated by a mountain range). This gives us more
flexibility than quasi unit disk graphs.

Formally, an \df{edge-labeling} of a graph $G=(V,E)$ is a map $\ell:E\rar\{N,F\}$. If $(u,v)\in E$, then $(u,v)$ is called \df{near} if $\ell(u,v)=N$,
and otherwise $(u,v)$ is called \df{far}.
In a unit disk (interval) representation $I$, each vertex $v\in V$ is represented as a disk (interval) centered at the point $I(v)\in \R^2$ ($\R$). We
denote by $||I(u)-I(v)||$ the distance between the points $I(u)$ and $I(v)$, and by
a slight abuse of notation, we also refer to $I(v)$ as the disk (interval) representing $v\in V$.
A weak unit disk (interval) representation of $G$ with respect to $\ell$ is
a representation $I$ such that for each edge $(u,v)\in E$, $||I(u)-I(v)||\leq d$ if and only if $\ell(u,v)=N$,
for some fixed unit $d>0$ (in other words, the disks and intervals have diameter $d$). Unless otherwise stated, we assume $d=1$.
In this paper, we study \df{weak unit disk (interval)} planar graphs, that is, planar graphs that have
appropriate representations for all possible edge-labelings.

\subsection{Related Work}
Weak unit disk graphs can be seen as a form of graph drawing/labeling
where closeness between vertices is used to define edges, albeit only from a defined set of
permissible edges. There have been many classes of graphs, defined on some notion
of closeness of the vertices.  \emph{Proximity graphs} are ones that
can be drawn in the plane such that every pair of adjacent vertices satisfies some fixed notion
of closeness, whereas every pair of non-adjacent vertices satisfy some notion of farness.
A common class of proximity graphs are \emph{Gabriel graphs},
first applied in~\cite{GS69} to the categorization of biological populations. Gabriel graphs can
be embedded in the plane so that, for every pair of vertices $(u,v)$, the disk with $u$ and $v$ as antipodal
points contains no other vertex if and only if $(u,v)$ is an edge. Recently, Evans et~al.~\cite{EGKLM12}
studied \emph{region of influence graphs}, where each
pair of vertices $u,v$ in the plane are assigned a region $R(u,v)$, and there is an edge
if and only if $R(u,v)$ contains no vertices except possibly $u$ and $v$. They generalized this
class of graphs to \emph{approximate proximity graphs}, where there are parameters
$\epsilon_1>0$ and $\epsilon_2>0$, such that a vertex other than $u$ or $v$ is contained in
$R(u,v)$, scaled by $1/(1+\epsilon_1)$ in an appropriate fashion, if and only if $(u,v)$ is an
edge, and the region $R(u,v)$, scaled by $1+\epsilon_2$ in an appropriate fashion, is empty
if and only if $(u,v)$ is not an edge. Such graphs place a stronger requirement on how far
away non-adjacent vertices must be than typical proximity graphs.

Weak unit ball representability in 1D is equivalent to \emph{threshold-coloring}~\cite{alam13}.
In this variant of the graph coloring, integer colors are assigned to the vertices
so that endpoints of near edges differ by less than a given threshold, while endpoints
of far edges differ by more than the threshold. It is shown that deciding whether a graph
is threshold-colorable with respect to a given partition of edges into near and far is equivalent
to the graph sandwich problem for proper-interval-representability, which is known to be
\NP-complete~\cite{Golumbic95}. Hence, deciding if a graph admits a weak unit interval
representation with respect to a given edge-labeling is also \NP-complete. Note the
difference with recognizing unit interval graphs, which
can be done in linear time~\cite{FMM95}. It is also known that planar graphs with girth
(the length of a shortest cycle in the graph) at least $10$ are always threshold-colorable.
Several Archimedean lattices (which correspond to tilings of the plane by
regular polygons), and some of their duals, the Laves lattices, are threshold-colorable~\cite{ourFunArxiv}.
Hence, these graph classes are weak unit interval graphs.

Unit interval graphs are also related to threshold and difference graphs.
In \emph{threshold graphs} there exists a real number $S$ and for every vertex
$v$ there is a real weight $a_v$ so that $(v,w)$ is an edge if and only if
$a_v + a_w \ge S$~\cite{Mahadev95}.
A graph is a \emph{difference graph} if there is a real number $S$ and for every vertex $v$
there is a real weight $a_v$ so that $|a_v| < S$ and $(v,w)$ is an edge if and only if
$|a_v - a_w| \ge S$~\cite{Hammer90}. Note that for
both these classes the existence of an edge is completely determined  by the threshold $S$, while
in our setting the edges defined by the threshold (size of the ball) must also belong to the original (not necessarily
complete) graph.
Threshold-colorability is also related to the \emph{integer distance graph} representation~\cite{Eggleton86,Ferrara05}. An integer distance graph is a graph with the set of integers as vertex set and with an edge joining two vertices $u$ and $v$ if and only if $|u - v| \in D$, where $D$ is a subset of the positive integers. Clearly, an integer distance graph is threshold-colorable if the set $D$ is a set of consecutive integers.

\subsection{Our Results}
We introduce the notion of weak unit disk and interval representations. While
finding representations with unit intervals is equivalent to threshold-coloring and so some results are already known, the
problem of weak unit disk representability is new. We first show that
recognizing weak unit disk graphs is hard:
For a graph $G$ with an edge-labeling $\ell$, it is \NP-hard to
decide if $G$ has a weak unit disk representation with respect to $\ell$, even if the
edges labeled $N$ induce a planar subgraph.
On the positive side, we show that any degree-2 contractible graph (as defined later)
has a weak unit disk representation. In particular, any series-parallel graph is
a weak unit disk representation.

We next study weak unit interval representations. It follows from~\cite{alam13} that
all planar graphs with high girth are always weak unit interval graphs.
We generalize the result by proving that graphs of bounded maximum average degree have
weak unit interval representations for any given edge-labeling.
In the other direction, we construct an example of a planar girth-4
graph which is not a weak unit interval, improving on the previously best girth-3 example.
Furthermore, we show that dense planar graphs do not always admit such a weak unit interval
graph representation.

Finally we study outerplanar graphs. It is known that some
outerplanar graphs with girth $3$ are not weak unit interval graphs,
and our example of a girth-4 graph is not outerplanar. Thus, a natural question
in this context is whether all girth-4 outerplanar graphs admit unit interval representation. We
answer the question in a positive way:
Every triangle-free outerplanar graph is a weak unit interval graph.

\section{Weak Unit Disk Graph Representations}
\label{sect:disks}

First we consider the complexity of recognizing weak unit disk graphs.

\begin{lemma}
\label{lem:nph}
For a graph $G$ with an edge-labeling $\ell$, it is \NP-hard to
decide if $G$ has a weak unit disk representation with respect to $\ell$, even if the
edges labeled $N$ induce a planar subgraph.
\end{lemma}

\begin{proof}
It is known that deciding if a planar graph is a unit disk graph is \NP-hard~\cite{BK98}.
Let $n$ be the number of
vertices of $G$, and define an edge-labeling $\ell$ of $K_n$ by setting $\ell(e) = N$ if
and only if $e$ is an edge of $G$. Clearly, a unit disk representation of $G$ is also a weak unit
disk representation of $K_n$ with respect to $\ell$ and vice versa.
\end{proof}

Note that the  argument above only proves \NP-hardness, and the problem
of deciding if a graph with an edge-labeling has a weak unit disk representation
is not known to be in \NP. The obvious approach is to use a
weak unit disk drawing as a polynomial size certificate. Unfortunately, it has recently been showed that
unit disks graphs on $n$ vertices may
require $2^{2^{\Theta(n)}}$ bits for a unit disk drawing with integer coordinates~\cite{MM13} .

\subsection*{Unit Disk Representation of Outerplanar and Related Graphs}

Note that the class of weak unit disk graphs strictly contains the class of weak unit interval graphs.
For example, in Fig.~\ref{fig:2dSungraph}, we give a representation of the sungraph, which is not a weak unit interval graph.
Our main goal here is to prove that every series-parallel graph is a weak unit disk graph, for every edge-labeling.
To this end, we will study a larger class of graphs, which we
call \df{degree-2 contractible} graphs.
A simple graph $G$ is a degree-2 contractible graph if one of the following holds:
\begin{enumerate}
	\item $G$ is an isolated vertex;
	\item each component of $G$ is a degree-2 contractible graph;
	\item $G$ has an edge $(u,v)$ such that $v$ has degree at most 2, and the graph
		obtained by contracting $(u,v)$ and removing parallel edges is a degree-2 contractible graph.
\end{enumerate}

\begin{figure}[tb]
\centering
  \subfigure[]{
  \includegraphics[scale=0.8,page=1]{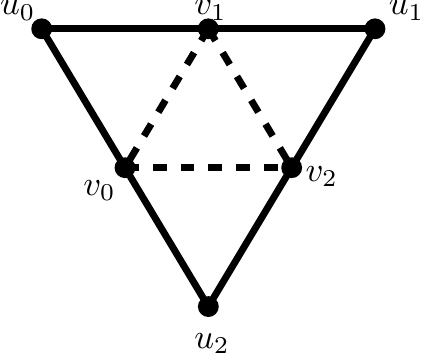}}
  ~~~~~~~~~~~~~~
  \subfigure[]{
  \includegraphics[scale=0.8,page=2]{sun_disks}}
  \caption{(a)~The sungraph has no weak unit interval representation, but
  (b)~it has a weak unit disk representation.
Near/far edges are indicated with solid/dashed line segments.}
  \label{fig:2dSungraph}
\end{figure}

\begin{theorem}
\label{thm:2dop}
Any degree-2 contractible graph is a weak unit disk graph.
\end{theorem}

\begin{proof}
Suppose that $G$ is a connected degree-2 contractible graph with $n>1$ vertices and edge-labeling
$\ell$. Assume that $G$ has no vertices of degree 1. We argue by
induction that $G$ is a weak unit disk graph. The base case is trivial, so consider the
following inductive hypothesis. Suppose that every degree-2 contractible graph $G'$ with
$n-1$ vertices has a weak unit disk representation $I'$ with respect to any edge-labeling
$\ell'$ of $G'$. Furthermore, suppose that (i)~the disks of the representation have diameter $2$,
(ii)~$I(v)$ has integer components for each vertex
$v$ of $G'$, and (iii)~for each edge $(u,v)$ of $G'$, $||I(u)-I(v)||\leq\sqrt{10}$.

Since $G$ is degree-2 contractible, $G$ has a vertex $v$ with exactly two neighbors $u$ and $w$,
such that contracting the edge $(u,v)$ results in a degree-2 contractible graph $G'$. We
adopt the convention that, instead of contracting $(u,v)$, we delete $v$ and add the
edge $(u,w)$ if it is not already present. By the inductive hypothesis, $G'$ has a weak unit disk representation
$I'$ with respect to the edge-labeling $\ell$ restricted to the edges of $G'$ (if edge $(u,w)$ does
not belong to $G$, give it an arbitrary label). Without loss of generality, we can assume that
$I'(u) = (0,0)$ and $I'(w) = (a,b)$ where $0\leq b<a$. We construct a representation $I$ of $G$
by setting $I(x) = I'(x)$ for every vertex $x\neq v$. To compute the value of $I(v)$, consider
Table~\ref{tab:2dop}. There we list for every possible value of $I(w)$, and every possible edge-labeling of $(u,v)$ and $(v,w)$ (except for one symmetric labeling), an appropriate value for $I(v)$
that satisfies the inductive hypothesis. The result follows.
\end{proof}

\begin{table}[tb]
    \centering
 \caption{Details of the proof of Theorem~\ref{thm:2dop}. For every
	edge-labeling (up to symmetry) and every possible value of $I'(w)$,
	we give a value for $I'(v)$. In the column for the edge-labeling, an empty
	cell indicates to take the value from the cell above.}
 \label{tab:2dop}
 \medskip
    \begin{tabular}{@{}cccc@{\hspace{1em}} cccc@{\hspace{1em}}cccc@{}}
        \toprule
        $\ell(u,v)$ & $\ell(v,w)$ & $I(w)$ & $I(v)$
        & $\ell(u,v)$ & $\ell(v,w)$ & $I(w)$ & $I(v)$
        & $\ell(u,v)$ & $\ell(v,w)$ & $I(w)$ & $I(v)$  \\
        \midrule
        $N$ & $N$ & $(1,0)$ & $(2,0)$   &   $N$ & $F$ & $(1,0)$ & $(0,2)$    & $F$ & $F$   & $(1,0)$ & $(2,2)$ \\
        & & $(2,0)$ & $(1,0)$           &&            & $(2,0)$ & $(0,1)$    &&            & $(2,0)$ & $(1,2)$\\
        & & $(3,0)$ & $(2,0)$           &&            & $(3,0)$ & $(0,1)$    &&            & $(3,0)$ & $(2,2)$\\
        & & $(4,0)$ & $(2,0)$           &&            & $(4,0)$ & $(1,0)$    &&            & $(4,0)$ & $(2,2)$\\
        & & $(1,1)$ & $(2,0)$           &&            & $(1,1)$ & $(-1,0)$   &&            & $(1,1)$ & $(0,3)$\\
        & & $(2,1)$ & $(2,0)$           &&            & $(2,1)$ & $(-1,0)$   &&            & $(2,1)$ & $(0,3)$\\
        & & $(3,1)$ & $(2,0)$           &&            & $(3,1)$ & $(1,0)$    &&            & $(3,1)$ & $(1,2)$\\
        & & $(2,2)$ & $(2,0)$           &&            & $(2,2)$ & $(1,0)$    &&            & $(2,2)$ & $(0,3)$\\

        \bottomrule
    \end{tabular}
\end{table}

Series-parallel graphs are defined as the graphs that do not have $K_4$ as a minor~\cite{D65}. Hence by definition,
these graphs are closed under edge contraction. It is also well-known that a series-parallel graph
has a vertex of degree 2, and that every outerplanar graph is a subgraph of a series parallel graph.
Thus, by Theorem~\ref{thm:2dop}, we have the following corollary.

\begin{corollary}\label{cor:opsp}
	Every outerplanar and series-parallel graph is a weak unit disk graph.
\end{corollary}

\section{Weak Unit Interval Graph Representations}
\label{sect:intervals}

In this section, we study weak unit interval representability, which is equivalent to
threshold graph coloring~\cite{alam13}. Given a graph $G=(V,E)$, an
edge-labeling $\ell:E\rar\{N,F\}$, and integers $r>0$, $t\ge0$, $G$ is said to be
\df{$(r,t)$-threshold-colorable} with respect to $\ell$ if there exists a coloring $c:V\rar\{1,\ldots,r\}$ such that for each edge $(u,v)\in E$,
$|c(u)-c(v)|\leq t$ if and only if $\ell(u,v) = N$. The coloring $c$ is known as a \df{threshold-coloring}.
It is easy to see that threshold-coloring is very similar to weak unit interval representation, with the only
difference being that weak unit interval graphs do not require integer coordinates. We show that this does not
matter.

\begin{lemma}\label{lem:equiv}
    A graph $G$ is a weak unit interval graph with respect to an edge-labeling $\ell$ if and only if
    $G$ is $(r,t)$-threshold-colorable with respect to $\ell$ for some integers $r>0$, $t\ge0$.
\end{lemma}

\begin{proof}
    Clearly, a threshold-coloring $c$ is a weak unit interval representation of $G$ with respect to $\ell$ (where
    we use $t$ as the unit of the representation), so we need only show that a weak unit interval
    representation of $G$ is equivalent to some threshold-coloring.
    Suppose that $I$ is a weak unit interval representation of $G$ with respect to $\ell$. If any of the
    intervals of $I$ intersect only at their endpoints, then we increase the length of each interval by some
    $\epsilon >0$, and choose $\epsilon$ so that the intervals have rational length. Next, we perturb the
    center point of each interval, in some fixed order, by some $\epsilon$ so that each interval is centered at
    a rational point. Next, we scale the representation so that the center of each interval is an integer,
    and the length of the intervals is an integer. The modified representation is a threshold-coloring (although $r$ and $t$ may be large).
\end{proof}

We now present a method for representing graphs, which
admit a decomposition into a forest and a 2-independent set.
By $G[U]$ we mean the subgraph of $G$ induced by the vertex set $U\subseteq V$.
Recall that a subset $\I$ of vertices in a graph $G$ is called
\df{independent} if $G[\I]$ has no edges. $\I$ is called \df{2-independent} if the shortest path in $G$ between any 2 vertices of $\I$ has length greater than 2. Similar decompositions have been applied to
other graph coloring problems~\cite{albertson04,timmons08,ourFunArxiv}.

\begin{lemma}
\label{lem:51coloring}
Suppose $G = (\I\cup \F,E)$ is a graph such that $\I$ is 2-independent, $G[\F]$
is a forest, and $\I\cap \F=\emptyset$. Then $G$ is a weak unit interval graph.
\end{lemma}

\begin{proof}
We assume that all the intervals are centered at integer coordinates and have length $d=1$.
Suppose $\ell:E\rar\{N,F\}$ is an edge-labeling.
For each $v\in \I$, set $I(v) = 0$. Each vertex in $G[\F]$ is assigned a point from
$\{-2,-1,1,2\}$ as follows. Choose a component $T$ of $G[\F]$, and select a root
vertex $w$ of $T$. If $w$ is far from a neighbor in $I$, set $I(w)=2$; otherwise, $I(w)=1$.
Now perform breadth first search on $T$, assigning an interval for each vertex as it is traversed. When we
reach a vertex $u\neq w$, it has one neighbor $t$ in $T$ which has been processed, and at most one neighbor
$v\in \I$. If $v$ exists, we choose the interval $I(u)=1$ if $\ell(u,v)=N$, and $I(u)=2$ otherwise. Then, if the edge-label
$(u,t)$ is not satisfied, we multiply $I(u)$ by $-1$.
If $v$ does not exist, we choose $I(u)= 1$ or $-1$ to satisfy the edge $(u,t)$.
By repeating the procedure on each component of $G[\F]$, we construct a representation of
$G$ with respect to the labeling $\ell$.
\end{proof}

Recall that the \df{maximum average degree} of a graph $G$, denoted $\mad(G)$,
is the average vertex degree of a subgraph $H$ with highest average degree.
It is known that every planar graph $G$ of maximum average degree $\mad(G)$ strictly
less than $\frac{26}{11}$ can be decomposed into a 2-independent set and a forest~\cite{bu09}.
Hence,

\begin{theorem}\label{thm:mad}
Every planar graph $G$ with $\mad(G) < \frac{26}{11}$ is a weak unit interval graph.
\end{theorem}

We also note that a planar graph with girth $g$ satisfies $\mad(G) < \frac{2g}{g-2}$. Therefore,
a planar graph with girth at least $13$ has always a weak unit interval representation.

\begin{wrapfigure}{t}{.38\textwidth}
    \centering
    \includegraphics[height=3cm]{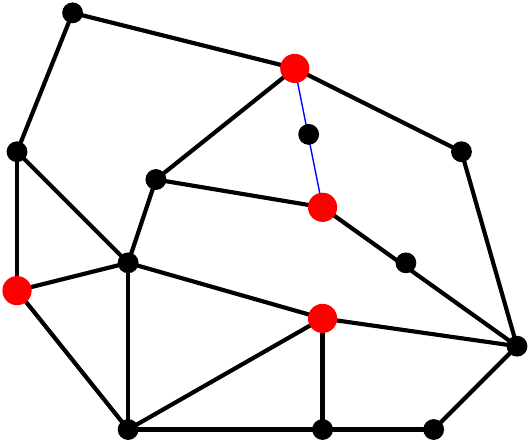}
    \caption{Decomposition of a graph into a nearly 2-independent set (red vertices)
    and a forest (black vertices and edges). Thin blue edges are \emph{bad}.}
    \label{fig:2ind}
  \vspace*{-4ex}
\end{wrapfigure}
Next we present a generalization of Lemma~\ref{lem:51coloring}, suitable for graphs which have an independent set that is in some sense nearly 2-independent.
The strategy is to delete certain edges so the independent set becomes 2-independent, obtain a unit interval representation using Lemma~\ref{lem:51coloring}, and
then modify it so that it is a representation of the original graph.
Formally, let $\I$ be an independent set in a graph $G$. Suppose that for every vertex $v\in \I$, there is at most one vertex $u\in \I$ such that the distance between
$v$ and $u$ in $G$ is 2. Also suppose that there is only one 2-path (a
path with 2 edges) connecting $v$ to $u$. Then we call $\I$ \df{nearly
 2-independent}. The pair $\{u,v\}$ is called an \df{$\I$-pair}, and
the edges of the path $(u,x,v)$ connecting $u$ and $v$ are called \df{bad edges},
which are associated with the bad pair $\{u,v\}$; see Fig.~\ref{fig:2ind}.

\begin{lemma}
\label{lem:multiplication}
Let $G=(\I\cup\F, E)$ be a graph such that $\I$ is a nearly 2-independent set, $G[\F]$ is a forest and $\I\cap\F=\emptyset$. Then $G$ is a weak unit interval graph
with respect to $\ell$.
\end{lemma}

\begin{proof}
Again, we assume that all the intervals are centered at integer coordinates. We use intervals of size $d=3$.

Suppose that $\ell:E\rar\{N,F\}$ is an edge-labeling of $G$. Let $E'\subseteq E$ be a set such that for each
$\I$-pair $\{u,v\}$, exactly one of the bad edges associated
with $\{u,v\}$ belongs to $E'$. Let $G'=(V,E-E')$. Then clearly $\I$ is a 2-independent set in $G'$, and $G'[\F]$ is a forest, so there exists a
weak unit disk representation $I'$ of $G'$ with respect to $\ell$.

We now modify $I'$ to construct a weak unit disk representation $I$ of $G$ with respect to
$\ell$. First, for each vertex $v\in V$, set $I(v)=0$ if
$I'(v)=0$, $I(v)=2$ if $I(v')=1$,
and $I(v)=5$ if $I'(v)=2$ (if $I'(v)$ is negative, do the same but set $I(v)$ negative).
It is clear that $I$ is a weak unit disk representation of $G'$. Now, let $(x,y)\in E'$.
One of these vertices, say $x$, is in $\I$ so $I(x)=0$, and $I(y)\in\{-5,-2,2,5\}$. Without
 loss of generality assume that $I(y)>0$; the case where $I(y)<0$ is symmetric.
Now it is possible that $\ell(x,y)=N$ but $||I(x)-I(y)||>3$ or that $\ell(x,y)=F$ but
$||I(x)-I(y)||\leq3$. In the first case, we must have $I(y)=5$. We modify $I$ so that $I(x)=1$
and $I(y)=4$.
Note that $y$ is still near to vertices with intervals centered at 2 or 5, and far from vertices with intervals
centered at less than 1. Similarly, $x$ is still close to the intervals at $-2$, $0$, or $2$, but far from $-5$
and $5$. Thus all the edges of $E-E'$ are satisfied by the modification of $I$, and
additionally the edge $(x,y)$ is satisfied.
In the second case, we have $I(y)=2$. We modify $I$ so that $I(x)=-1$ and $I(y)=3$. As
before, no edges which disagreed with the edge-labeling still disagree with the
edge-labeling.

Since $\I$ is nearly 2-independent,
our modifications to the representation $I$ will not affect non-local vertices, as every vertex in $\I$ is
adjacent to at most one edge of $E'$.
\end{proof}

\subsection*{Weak Unit Interval Representation of Outerplanar Graphs}

It is known that some outerplanar graphs containing triangles are not weak unit interval graphs, e.g., the sungraph
in Fig.~\ref{fig:2dSungraph}. Hence, we study weak unit interval representability of triangle-free outerplanar
graphs. We start with a claim for graphs with girth 5.

\begin{lemma}
\label{thm:girth5}
An outerplanar graph with girth 5 is a weak unit interval graph.
\end{lemma}

\begin{proof}
We prove that girth-5 outerplanar graphs may be decomposed into a forest and a 2-independent set using induction on the number of internal faces. The result will follow from Lemma~\ref{lem:51coloring}. The claim is trivial
for a single face, so assume that it is true for all girth-5 outerplanar graphs with $k\geq1$ internal faces. Let $G$ be a girth-5 outerplanar graph with $k+1$ internal faces. Since $G$ is outerplanar, it must have at least one
face $f=(v_1,\dots,v_l)$, $l\geq5$, such that every vertex of $f$ except $v_1,v_l$ is of degree 2. Consider the graph $G'$ obtained by deleting $v_2,\dots,v_{l-1}$. $G'$ has a decomposition into
a 2-independent set $\I$ and a forest $T$. Now we will add the vertices $v_2,\dots,v_{l-1}$ to either $\I$ or $T$ so that $\I$ is a 2-independent set in $G$, and $T$ is a forest. If either of $v_1,v_l$ belongs to $\I$, then
add all the remaining vertices to $T$. Otherwise, add $v_3$ to $\I$ and the rest to $T$. Since $v_1,v_l$ are not in $\I$, $v_3$ has distance at least 3 from any other element of $\I$.
\end{proof}

Next our goal is to show that a triangle-free outerplanar graph $G$ always has a weak
unit interval representation for any edge-labeling.
We assume that all the intervals are centered at integer coordinates and we use intervals of size $d=2$.
Our strategy is to find a representation of $G$ by a traversal of its
weak dual graph $G^*$ (the planar dual minus the outerface), where we find intervals for all the vertices
in each interior face of $G$ as it is traversed in $G^*$. Since we are considering
triangle-free graphs, this implies that we take
a path $P_n= (u_1,u_2,\ldots,u_n)$, $n\ge 4$, where the two end vertices $u_1$ and $u_n$
are already processed and we need to assign unit intervals
to the internal vertices $u_2$, \ldots, $u_{n-1}$ of $P_n$. We additionally maintain the invariant in our
representation that for each edge $(u,v)$ of $G$, $||I(u)-I(v)||\leq6$. For a particular edge-labeling
$\ell$ of $P_n = (u_1, \ldots, u_n)$, call a pair of coordinates
$x$, $y$ \df{feasible} if there is a weak unit disk representation $I$ of $P_n$ for $\ell$ with $d=2$,
where $I(u_1)=x$, $I(u_n)=y$, and for any
$i\in\{1,\ldots,n-1\}$, $||I(u_i)-I(u_{i+1})||\le 6$.
We first need the following three claims.

\begin{pp}
\label{pp:1}
For any value of $x\in\{2,3,-2,-3\}$, the pair $0,x$ is feasible for any
 edge-labeling $\ell$ of $P_3=(u_1,u_2,u_3)$.
\end{pp}
\begin{proof} Without loss of generality, we may assume that $x>0$. We compute a desired
 weak unit disk representation $I$ with $r=2$ for $P_3$ with respect to $\ell$ as follows.
 Assign $I(u_1)=0$ and $I(u_3)=x$. Assign $I(u_2)$ in such a way that $|I(u_2)|=2$ if
 $\ell(u_1,u_2)=N$, and $|I(u_2)|=3$ if $\ell(u_1,u_2)=F$. Then choose the sign of $I(u_2)$ to be
 the same as $I(u_3)$ if $\ell(u_2,u_3)=N$, and the opposite of $I(u_3)$ if $\ell(u_2,u_3)=F$.
\end{proof}

\begin{pp}
\label{pp:2}
For any edge-labeling of $P_3=(u_1,u_2,u_3)$, either $0$, $4$ or $0$, $6$ are feasible.
\end{pp}
\begin{proof} We compute a desired weak unit disk representation $I$ with $r=2$ for $\ell$ as follows.
 If $\ell(u_1,u_2)=l(u_2,u_3)=N$, then $I(u_1)=0$, $I(u_2)=2$, and $I(u_3)=4$. Otherwise,
 assign $I(u_1)=0$, $I(u_3)=6$, and $I(u_2)=2,3$ or $4$ in case the pair of edge-labelings
 $(l(u_1,u_2), l(u_2,u_3))$ have values $(N,F)$, $(F,F)$, and $(F,N)$, respectively.
\end{proof}

\begin{pp}
\label{pp:3}
For any value of $x\in[-6,6]$, the pair $0,x$ is feasible for any edge-labeling of
 $P_n=(u_1,u_2,\ldots,u_n)$, $n\ge 4$.
\end{pp}
\begin{proof} Without loss of generality, let $x\ge 0$. Consider first the case for
 $n=4$. Take a particular edge-labeling $\ell$ of $P_4$. For any value of $0\le x\le 5$, there
 is at least one number $y\in\{2,3,-2,-3\}$ and at least one number $z\in\{2,3,-2,-3\}$ such
 that $|x-y|\le 2$ and $2<|x-z|\le 6$. In particular, it suffices to choose for $x=0$, $y=2$,
 $z=3$; for $x=1,2,3,4$, $y=2$, $z=-2$ and for $x=5$, $y=3$, $z=2$. Thus if
 $0\le I(u_4)\le 5$, and regardless of whether $\ell(u_3,u_4)$ is $N$ or $F$, one can choose a
 value for $I(u_3)$ from $\{2,3,-2,-3\}$ respecting both the edge-labeling of $(u_3,u_4)$
 and the property that $||I(u_3)-I(u_4)||\le 6$. Then by~\ref{pp:1}, $0$ and $x$
 is feasible for the edge-labeling $\ell$ of $P_4$. A similar argument shows that if $\ell(u_3,u_4)=F$,
 then $0$ and $x=6$ is feasible. On the other hand, if $x=6$ and $\ell(u_3,u_4)=N$, then both
 $4$ and $6$ are valid choices for $I(u_3)$. Then by~\ref{pp:2}, $0$ and $6$ is feasible
 for any edge-labeling $\ell$ of $P_4$.

Consider now the case with $n>4$. Then assign coordinates $I(u_1)=0$, $I(u_n)=x$ and for
 $i\in\{n-1, \ldots, 4\}$, assign $I(u_i)\in[-6,6]$ such that it respects both $\ell(u_i,u_{i+1})$
 and the property that $||I(u_i)-I(u_{i+1})||\le 6$. Then a similar argument as that for $n=4$ can be used to extend this representation to $u_2$ and $u_3$.
\end{proof}

The next corollary immediately follows from~\ref{pp:3}.

\begin{corollary}
\label{cor:trifree}
 Any pair $x$, $y$ with $|x-y|\le6$, is feasible for any edge-labeling of
 $P_n=(u_1,u_2,\ldots,u_n)$, $n\ge 4$.
\end{corollary}

\begin{theorem}
Every triangle-free outerplanar graph is a weak unit interval graph.
\end{theorem}

\begin{proof}
    If $G$ is not 2-connected, we augment it in the following way. If
 $G$ has a bridge $(v,w)$, let $u\neq w$ be a neighbor of $v$, and $x\neq v$ a neighbor
 of $w$, then add the edge $(u,x)$. If $G$ has a cut vertex $v$, then let $H_1,\dots,H_k$
 be the 2-connected components of $G$ containing $v$. For $i\in\{1,\dots,k-1\}$, let
 $u$ be a neighbor of $v$ in $H_i$, and $w$ be a neighbor of $v$ in $H_{i+1}$.
 Add the path $(u,x,w)$, where $x$ is a new vertex. Clearly, any weak unit interval representation of the
 new 2-connected graph is also a weak unit interval representation of $G$, and the new graph is
 outerplanar with girth 4.

Now let $G$ be a 2-connected triangle-free outerplanar graph with $n>4$ vertices embedded in the plane with every
 vertex on the outerface, and let $\ell$ be an edge-labeling of $G$. We next compute a
 weak unit interval representation of $G$ for $\ell$. The proof is by induction, with the $n$-vertex cycle as a base case. Assume the inductive hypothesis that every triangle-free outerplanar graph with fewer than $n$
 vertices is a weak unit interval graph. Further, assume that for such a graph $G'$ and any edge-labeling $\ell'$
 of $G'$, there is a weak unit interval representation of $G'$ for $\ell'$ such that any 2 neighbor vertices $u$ and $v$
 satisfy $||I(u)-I(v)||\le6$. It is clear that if $G$ has at least two cycles, then $G$ has a path $P_k=(u_1,\ldots,u_k)$ where $k\ge4$ such that $\deg(u_i)=2$ for some $1<i<k$. The theorem follows from the inductive hypothesis and Corollary~\ref{cor:trifree}.
\end{proof}

\subsection*{Planar Graphs without Weak Unit Interval Representations}

Planar graphs with high edge density may not have weak unit interval representations.
First we prove the result for a \df{wheel graph}, defined as $W_n$, $n\geq4$,
formed by adding an edge from a vertex $v_1$ to every vertex of an $(n-1)$-cycle
$(v_2,\dots,v_{n},v_2)$.

\begin{lemma}
\label{lm:wheel}
A wheel graph is not a weak unit interval graph.
\end{lemma}

\begin{proof}
Define an edge-labeling $\ell$ of $W_n$ by $\ell(v_2,v_n) = F$, $\ell(v_1,v_i)=F$
for $3 \le i \le n-1$, and every other edge labeled $N$; see Fig.~\ref{fig:wheel}. Suppose $I$ is a
weak unit interval representation of $W_n$ with respect to $\ell$.
Since only one edge of the triangle $(v_1,v_2,v_n,v_1)$ is far, it
must be that $I(v_1)\neq I(v_2)$, so we may assume that $I(v_1)<I(v_2)$. For $3 \le i \le n$, if
$I(v_{i-1})>I(v_1)$, we have $I(v_i)>I(v_1)$, since $\ell(v_{i-1},v_i)=N$ and either
$\ell(v_1,v_{i-1})$ or $\ell(v_1,v_i)$ is $F$. But then $I(v_1)<I(v_2)\leq I(v_1)+1$, and
$I(v_1)< I(v_n)\leq I(v_1)+1$, contradicting the fact that $\ell(v_2,v_n)=F$ and $I$ is a
weak unit interval representation.
\end{proof}

\begin{figure}[t]
	\centering
	\subfigure[][] {
		\includegraphics{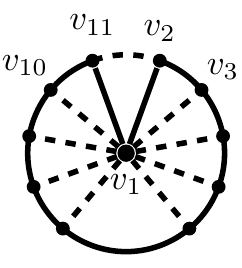}
		\label{fig:wheel}
	}
~~~~~~~~~~~~~~~~~~~~~~~~~~~~
	\subfigure[][] {
		\includegraphics{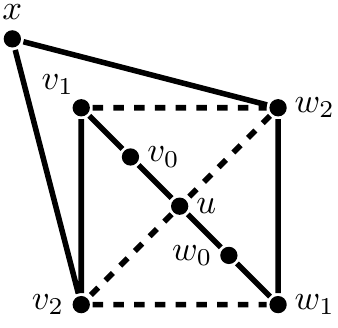}	
		\label{fig:girth4cex}
	}
	\caption{
(a)~A wheel graph $W_{11}$ with an edge-labeling, which does not have a weak unit interval representation.
(b)~A girth-4 graph with an edge-labeling, which does not have a weak unit interval representation.
}
\end{figure}

Using Lemma~\ref{lm:wheel}, it is easy to see that any maximal planar graph with $|V|\ge 4$
is not a weak unit interval graph.
Indeed, consider such a graph $G = (V, E)$ and a vertex $v\in V$; the neighborhood
$N(v) = \{u \,|\, (v, u)\in E\}$ together with $v$ induces a wheel subgraph. The observation
leads to the following theorem.

\begin{theorem}\label{thm:maximal}
Every planar graph $G$ with $\mad(G) \ge \frac{11}{2}$ is not
a weak unit interval graph.
\end{theorem}

\begin{proof}
To prove the claim, we show that a weak unit interval planar graph has at most
$\lfloor 11|V|/4\rfloor-6$ edges.

Consider a vertex $v$ of a weak unit interval planar graph $G=(V, E)$ and assume
it is embedded in the plane.
The neighborhood of $v$ is acyclic; otherwise $v$ and its neighborhood induce a wheel, which by
 Lemma~\ref{lm:wheel} is not a weak unit interval graph. Thus the number of edges
 between any two neighbors of $v$ is at most $\deg(v)-1$, where $\deg(v)$ is the degree of $v$.
Summing over all vertices, we get $S= 2|E| - |V|$. Let $T$ and
 $\overline{T}$ be the sets of triangular and non-triangular faces in an embedding
 of $G$. For each triangle $t \in T$ each of the edges in $t$ is counted once
 in $S$. Thus, $2|E|-|V|\ge 3|T|$ $\Rightarrow |T|\le (2|E|-|V|)/3$.
Counting both sides of the edges we get $2|E|\ge 3|T|+4|\overline{T}|$
 $\Rightarrow |T| + |\overline{T}| \le (2|E|+|T|)/4 \le (8|E|-|V|)/12$, since
 $|T|\le (2|E|-|V|)/3$. Thus, from Euler's formula $|V|-|E|+|T|+|\overline{T}|=2$, we have
 $|V|-|E| + (8|E|-|V|)/12 \ge 2$ $\Rightarrow |E| \le 11|V|/4 -6$.
\end{proof}

In~\cite{alam13} all examples of graphs without threshold-coloring (and thus, not weak unit interval graphs)
have girth 3. We strengthen the bound by proving the following.

\begin{lemma}
There exist planar girth-4 graphs that are not weak unit interval graphs.
\end{lemma}

\begin{proof}
Consider the graph in Fig.~\ref{fig:girth4cex}. Suppose there exists a weak unit interval
representation $I$. Without loss of generality suppose that $I(w_2) > I(u)$.
Let us consider two cases.  First, suppose $I(v_2) < I(u)$. Since the edges $(u,v_2)$
and $(u,w_2)$ are labeled $F$, it must be that $I(v_2) < I(u)-1$ and $I(u)+1 < I(w_2)$. Then
vertex $x$ must be represented by an interval near to both of these, which is impossible since
$||I(v_2)-I(w_2)|| > 2$.

Second, suppose $I(v_2) > I(u)$. Then $I(v_1)\geq I(v_2)-1>I(u)$, and
$I(u)<I(w_2)$ implies that $I(v_1)<I(w_2)$. Similarly, $I(w_1)<I(v_2)$. Now, either
$I(w_2)\leq I(v_2)$, or $I(v_2)<I(w_2)$ . In the first case, $w_2$ is
near to $v_1$ since $I(v_1)<I(w_2)\leq I(v_2)$ and $||I(v_1)-I(v_2)||\leq 1$, a contradiction.
The second case leads to a similar contradiction.
\end{proof}

\section{Conclusion and Open Problems}
We studied weak unit disk and the weak unit interval representations
for planar and outerplanar graphs. Many interesting open problems remain.

\begin{enumerate}
\item Deciding whether a graph
is a weak unit disk (interval) graph with respect to a given edge-labeling is \NP-complete. However, the
problem of deciding whether a (planar) graph is a weak unit disk (interval) graph is open.
Note that the class of weak unit disk (interval) graphs is not closed under minors, as
subdividing each edge of a planar graph three times results in a planar graph with girth at
least $10$, which is a weak unit interval graph.

\item
 Tightening the lower and upper
bounds for maximum average degree of
weak unit interval graphs, given in Theorems~\ref{thm:mad}
and~\ref{thm:maximal}, is a challenging open problem.
Based on extensive computer experiments, we conjecture that there
are no weak unit interval graphs with more than $2|V|-3$ edges.

\item We considered planar graphs,
but little is known for general graphs. In particular, it
would be interesting to find out whether the edge density of
weak unit disk (interval) graphs is always bounded by a constant.
\end{enumerate}

\noindent{\bf Acknowledgments:} We thank Michalis Bekos, Gasper
Fijavz, and Michael Kaufmann for productive discussions about the problems.

\begin{small}
\bibliographystyle{splncs03}
\bibliography{literature}
\end{small}

\newpage
\appendix

\section*{\LARGE Appendix}

\section{Unit-Cube Contact Representations of Planar Graphs}
\label{app:cubes}
This type of coloring is motivated by the geometric problem of
\emph{unit-cube contact representation} of planar
graphs~\cite{Bremner12}.
In such a representation, each vertex of a graph is represented by a
 unit-size cube and each edge is realized by a common boundary with
non-zero area between the two corresponding cubes.

The problems are related, as threshold-coloring can be
used to find unit-cube contact representations. In particular if a
graph $G$ has a unit-square contact representation $\Gamma$ in the plane and if $H$ is
any subgraph of $G$ formed by the
near edges of $G$ in some threshold-coloring, then one can find a unit-
cube-contact representation of $H$.

Formally, let $H$ be a threshold subgraph of $G=(V,E)$ and let
$c:V\rightarrow \{1\dots r\}$ be an $(r,t)$-threshold-coloring of $G$ with respect to an edge-labeling
$\ell$. We now compute a unit-cube contact representation of $H$ from
$\Gamma$ using $c$.
Assume (after possible rotation and translation) that the bottom face for each
cube in $\Gamma$ is co-planar with the plane $z=0$; see Fig.~\ref{fig:teaser}(a).
Also assume (after possible scaling)
that each cube in $\Gamma$ has side length $t+\epsilon$, where $0<\epsilon<1$.
Then we can obtain a unit-cube contact representation of $H$ from $\Gamma$ by lifting
the cube for each vertex $v$ by an amount $c(v)$ so that its bottom face is at $z=c(v)$; see
Fig.~\ref{fig:teaser}(b).
Note that for any edge $(u,v)$ with $\ell(u,v)=N$, the relative distance between the bottom faces of the
cubes for $u$ and $v$ is $|c(u)-c(v)|\le t<(t+\epsilon)$; thus the two cubes maintain contact. On the other hand, for each pair of vertices $u,v$ with
$\ell(u,v)\neq N$, one of the following two cases occurs: (i)~either $(u,v)\notin E$ and
their corresponding cubes remain non-adjacent as they were in $\Gamma$; or
(ii) $\ell(u,v)=F$ and the relative distance between the bottom faces of the two
cubes is $|c(u)-c(v)|\ge (t+1)>(t+\epsilon)$, making them non-adjacent.

\begin{figure}[h]
\centering
\includegraphics[width=0.9\textwidth]{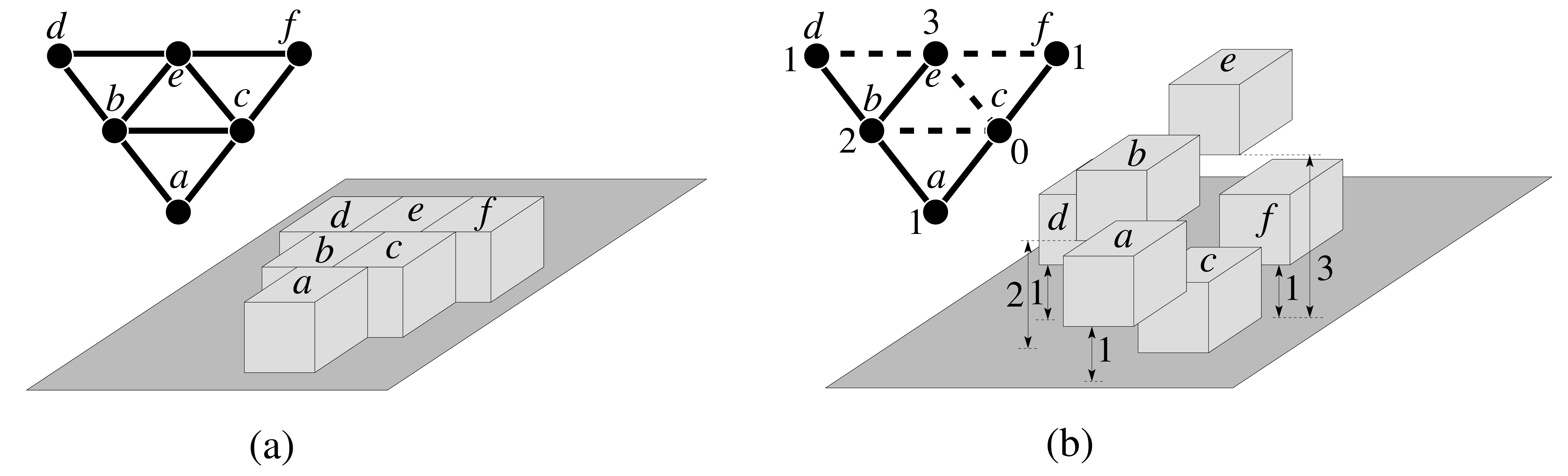}
\caption{(a) A graph $G$ with a unit-square contact representation. (b)~A threshold-coloring of $G$
 with threshold $1$ and a unit-cube-contact representation of the subgraph formed by near
 edges~\cite{alam13}. The vertex colors correspond to the elevation of the cubes, while
 the threshold gives the size of the cubes.
 In the figure we indicate $N$ edges with solid lines and $F$ edges with dashed lines.}
\label{fig:teaser}
\end{figure}

\end{document}